\documentclass[review]{elsarticle}

\usepackage{amssymb}
\usepackage[figuresright]{rotating}
\usepackage{multirow}
\usepackage{algorithm}
\usepackage{algorithmicx}
\usepackage{algpseudocode}
\usepackage{subfigure}
\usepackage{amsmath}
\usepackage{amsthm}
\usepackage{amsfonts}
\usepackage{bbm}
\usepackage{color}
\usepackage{url}
\usepackage{booktabs}
\usepackage{ulem}

\newtheorem{theorem}{Theorem}

\newtheorem{definition}{Definition}
\newtheorem{example}{Example}
\newtheorem{proposition}{proposition}

\usepackage{lineno}

\journal{Information Sciences}

\begin{document}

\begin{frontmatter}



\title{Partition-based differentially private synthetic data generation}


\author[label1]{Meifan Zhang}
\author[label1]{Dihang Deng}
\author[label1]{Lihua Yin$^*$}
\cortext[mycorrespondingauthor]{Corresponding author}
\ead[url]{yinlh@gzhu.edu.cn}

\affiliation[label1]{organization={Cyberspace Institute of Advanced Technology},
            addressline={Guangzhou University},
            city={Guangzhou},
            postcode={510006},
            state={},
            country={China}}

%

\begin{abstract}
  Private synthetic data sharing is preferred as it keeps the distribution and nuances of original data compared to summary statistics. The state-of-the-art methods adopt a select-measure-generate paradigm, but measuring large domain marginals still results in much error and allocating privacy budget iteratively is still difficult. To address these issues, our method employs a partition-based approach that effectively reduces errors and improves the quality of synthetic data, even with a limited privacy budget. Results from our experiments demonstrate the superiority of our method over existing approaches. The synthetic data produced using our approach exhibits improved quality and utility, making it a preferable choice for private synthetic data sharing.
  \end{abstract}



\end{frontmatter}



\section{Introduction}\label{sec:introduction}

Data sharing is crucial for fully utilize data, but it is vital to prioritize user privacy. To protect data privacy, many experts advocate the use of differential privacy, a well-recognized technique that involves introducing noise to statistical results~\cite{DBLP:conf/tcc/DworkMNS06, Dwork2006OurDO, Dwork2010BoostingAD}. However, summary statistics alone may not suffice for all analysis needs. Synthetic data offers a compelling alternative for data sharing as it preserves the original data's structure, distribution, and intricacies. It can seamlessly accommodate any analysis task originally intended for the source data.


There are currently two main categories of methods for generating synthetic data: (1) GAN-based methods that introduce noise to gradients, and (2) marginal-based methods that follow a \textbf{select-measure-generate} paradigm.
Previous research has indicated that marginal-based methods tend to be more accurate for generating synthetic data compared to GAN-based methods~\cite{Jordon2018PATEGANGS, DBLP:conf/nips/ChenOF20, DBLP:journals/corr/abs-2011-05537}, and recent studies have primarily focused on this approach~\cite{McKenna2022AIMAA, Cai2021DataSV, McKenna2021WinningTN}.
The basic idea of a marginal-based method is to (i) select some low-dimensional marginals, (ii) measure these marginals under DP, and (iii) generate synthetic data that aligns with the noisy measurements. Within the select-measure-generate framework of marginal-based methods, the ``generate'' step can be effectively accomplished using private-PGM~\cite{McKenna2019GraphicalmodelBE}, which learns the distribution from the noisy measurements according to Probabilistic Graphical Model (PGM). Consequently, the quality of the synthetic data heavily depends on the effectiveness of the ``select'' and ``measure'' steps.

The ``select'' step aims to find some important marginals from the workload, which can represent the entire distribution. The ``measure'' step aims to add noise to the selected marginals to satisfy DP. Each noisy marginal costs a portion of the privacy budget. Allocating a larger budget results in a reduced level of noise introduced by DP. If the budget is sufficient, we can afford using a large number of marginals to capture the characteristic of the distribution while introducing a small amount of noise to each marginal. Thus, it can approximate the distribution well. However, when the privacy budget is limited, both the ``select'' and ``measure'' strategies are non-trivial.

On one hand, selecting marginals that capture important characteristics of the original data within a limited privacy budget is a challenging task.
It is not feasible to use all marginals as measurements since each measurement needs to have noise added to it to protect privacy, and this noise incurs a cost in terms of the privacy budget. Therefore, the ``select'' part of the process aims to choose marginals that make the maximum contribution in capturing the data distribution while staying within the constraints of the limited privacy budget. Existing methods like AIM~\cite{McKenna2022AIMAA} and MWEM+PGM~\cite{DBLP:conf/nips/HardtLM12} take an iterative approach to select marginals that are poorly approximated by the current distribution. Because these marginals tends to contribute more to learining the distribution. But they may waste some privacy budget on some marginals bringing little contribution without proper privacy allocation methods. Another approach PrivSyn~\cite{Zhang2020PrivSynDP} employs a greedy selection strategy select all the marginals before measuring them. It aims to minimize the error by considering both the noise of the chosen marginals and the correlation score of the unchosen marginals. This approach ensures that the selected marginals not only have low noise but also maintain important correlation of different dimensions. However, since the value of some marginals may be covered by others, PrivSyn may waste some budget for the duplicate marginals.
Thus, we need to select marginals bringing important information of the distribution while requiring little privacy.

 \begin{figure*}[htbp]
\centering
\includegraphics[scale=0.5]{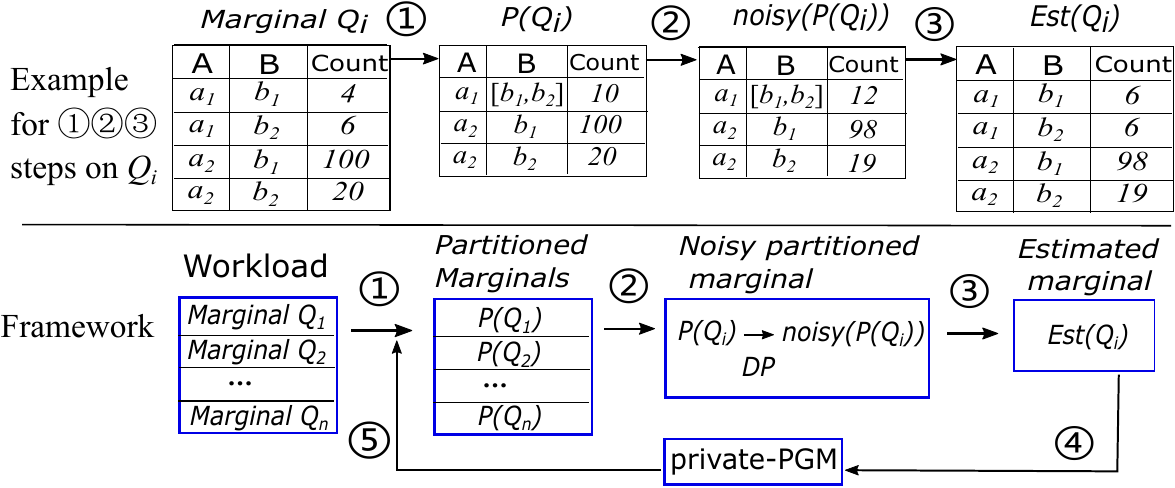}
\caption{Framework}
\label{Fig:framework}
\end{figure*}

On the other hand, determining the best strategy for allocating the limited privacy budget while selecting the marginals iteratively is also not straightforward. While it may seem logical to allocate more budget for larger marginals, this approach has its trade-offs. Since the privacy budget is limited, allocating more budget to one task reduces the budget available for others. Consequently, this can impact the number of measurements that can be taken, potentially affecting the overall accuracy of the data analysis. PrivSyn provides a method for privacy allocation for different margnials, but it only works on the selected marginals.  Allocating privacy budget iteratively is still difficult, MWEM allocate a uniform budget for each chosen marginal, while AIM starts with a fixed budget and increases it once the current budget is not enough to choose a new marginal that contribute positively to the distribution model. These methods do not guarantee that each marginal contributes sufficiently to learning the data distribution. In scenarios where the total budget is limited, even with optimal budget allocation, each marginal receives only a small portion of the budget, resulting in a significant amount of noise. Therefore, instead of focusing on finding an optimal budget allocation method, we need to fundamentally reduce the privacy budget required for each marginal when the total budget is insufficient.

To address the first challenge, for selecting a valuable marginal in each iteration, we propose a method considering the contribution per unit privacy budget as a metric for selecting an appropriate marginal, and it also ensures that the chosen marginal provides a sufficient utility of contribution. 
At first glance, it may seem simple to set the privacy budget first and then identify the marginal allocation with the highest potential contribution under that budget. However, it is important to note that the optimal selection may change if the budget is modified, and how to set the budget is not clear. Additionally, following a pre-set privacy budget does not guarantee that the chosen marginal allocation will bring sufficient contribution if all candidate marginals contribute minimally under that budget. To ensure that each noisy marginal brings sufficient contribution, we use a threshold to constrain the impact of noise on the contribution. We then compute the required privacy budget that satisfies this limit.

To address the second challenge, for reducing the error introduced by noisy measurement, we propose an adaptive partitioning approach to reduce the size of each marginal. While adding noise to each cell of the marginals is the standard method to achieve differential privacy, large margnials with a large number of cells bring significant noise error. These large marginals consume a large amount of privacy budget to ensure the untility of measurements, otherwise, they cost privacy budget without contributing much, or even negatively impacting the distribution. Merging uniformly distributed small cells into larger ones is a strategy to reduce the size of each large marginal. The count for each cell in a large partition can be estimated based on the assumption of a uniform distribution. Partition reduce the noise error but introduce the reconstruction error when the uniform distribution assumption is not true. Since our proposed selection strategy limits the utility of contribution, the only remaining task is to find the optimal partition that minimizes the required privacy budget while satisfying this constraint.

We proposed a framework shown in figure~\ref{Fig:framework}, which contains five main steps to learn a distribution that best explains the noisy measurements: \textcircled{1} compute a proper partition for each marginal in the workload according to the current distribution approximated by the private-PGM, \textcircled{2} select one partitioned marginal bringing the maximum potential contribution to reducing the approximated distribution error, \textcircled{3} add noises to the cells of the chosen partitioned marginal, estimate each value of original marginal according to the uniformly distributed assumption and include the noisy marginal in measurements, \textcircled{4} update the private-PGM model by learning the distribution well explains the measurements, \textcircled{5} adopt the updated model to select a new measurement in the next iteration.
The main concept is to decrease the size of each marginal, thereby reducing the privacy budget required. This approach also ensures that the contribution from each marginal remains valuable even after introducing errors caused by differential privacy and partitioning. By reducing the privacy cost while maintaining similar utility of each marginal, we can allocate more budget for more marginals to capture additional characteristics of the distribution. This leads to an increase in the overall utility of the synthetic data.

The main contributions of this work are summarized as follows:
\begin{itemize}
\item We propose a partition-based differentially private synthetic data generation method (PPSyn). We propose marginal partition algorithms for both one-dimensional and multidimensional marginals, which reduces the size of each marginal without bringing much reconstruction error.
\item We proved that each selected marginal brings a portion of positive contribution to improve the distribution model. 
\item We conduct extensive experiments on real-world datasets. The experimental results show that the proposed method outperforms the state-of-the-art workload-based methods on different data analysis tasks including range query and data classification.
\end{itemize}

The remainder of this paper is organized as follows. In section 2, we survey the related work for this paper. In section 3, we introduce some preliminaries. Section 4 presents the framework of the proposed method PPSyn. In section 5, the experimental results show the performance of the proposed algorithm. In section 6, we conclude the paper with future directions.

\section{Related work}
Differential privacy~\cite{Dwork2006DifferentialP}, has garnered widespread attention and has now become the de facto notion in privacy protection. Researchers have provided numerous studies and theories about DP algorithms~\cite{TCS-042, DBLP:books/sp/17/Vadhan17, DBLP:series/synthesis/2016Li}, but most of them are implemented for specific tasks~\cite{DBLP:conf/ccs/ZhangWLHC18,DBLP:conf/icde/YangCSCRL19,McKenna2018OptimizingEO}.
The difference is that synthetic data has a high degree of concern in the field of differential privacy because of its suitable for all data processing workflows designed for raw data. The field of differential privacy synthetic data has contributed a lot of theoretical research work for query release~\cite{DBLP:conf/nips/HardtLM12, Zhu2017DifferentialPA, McKenna2018OptimizingEO, Li2016DifferentialPF, McKenna2022AIMAA}.
These methods, such as MWEM~\cite{DBLP:conf/nips/HardtLM12}, although providing theoretically optimal error guarantees, have an exponential dependence on the data demension in terms of computational time overhead. Moreover, this exponential running time is known to be necessary in the worst case. Thus, researchers have made many different attempts.

\textbf{Marginal-based Methods.} These methods aim to derive an approximate distribution of the original data by using low-dimensional marginal information combined with generative modeling methods such as probabilistic graphs. They can be categorized into two types. (1) Workload-based methods: These methods use the workload derived from user-contributed query items to seek a query distribution that preserves statistical characteristics. For example, Gaboardi et al.~\cite{Gaboardi2014DualQP} introduced dual representation for optimizing data distribution, and later, Vietri et al.~\cite{Vietri2020NewOA} further analyzed and solved the optimization problem using integer program solvers. RAP~\cite{DBLP:conf/icml/AydoreBKKM0S21} utilizes a continuous relaxation of the Projection Mechanism and combines existing tools including autodifferentiation to find the best synthetic dataset. (2) Data-based methods: These methods construct low-dimensional noise marginals from data. For instance, PrivBayes~\cite{Zhang2014PrivBayesPD} estimates synthetic data based on measured marginals using Bayesian networks. It selects valuable measured marginals based on mutual information criteria to build the Bayesian networks. PGM~\cite{McKenna2019GraphicalmodelBE} effectively utilizes graphical models to reveal high-level information about a data distribution through noisy measurements. Many differential privacy synthesis data methods based on PGM have emerged since then. These methods follow a select-measure-generate paradigm but vary in their operation at different stages. PrivMRF~\cite{Cai2021DataSV}  selects reasonable constraints to simplify the graphical model's size and chooses candidate noisy measurements through Markov Random Field. It also builds a compact representation of reasonable joint probability distributions to avoid non-existent conditional assumptions in the graphical model.

\textbf{Deep Generative Model Based Methods.} To address the need for combining differential privacy (DP) with machine learning, several methods based on Generative Adversarial Networks (GAN) have been proposed~\cite{DBLP:conf/pkdd/AbayZKTS18, DBLP:journals/corr/abs-1802-06739, DBLP:journals/isci/WangCSW23,DBLP:conf/dsaa/FaisalMLW22, DBLP:journals/corr/abs-1801-01594}. The main idea behind thses methods is to add noise, thereby achieving privacy protection. However, the addition of noise can lead to gradient explosion, which requires gradient clipping to prevent instability. Unfortunately, gradient clipping can negatively impact the accuracy of the model, resulting in reduced performance. As a result, GAN-based methods currently have limitations and may not perform as well as other approaches.

The selection of the ideal set of measurement candidates is crucial for those marginal-based methods. As AIM~\cite{McKenna2022AIMAA} has done, the ideal joint probability distribution can be well constructed by continuously selecting the marginals that are poorly approximated according to the current model. However, the selected marginals often tend to be small in size, since large marginals containing more cells brings more noise and potentially distort the true distribution. 
To overcome these limitations, we propose partition methods that aim to reduce the accumulation of unnecessary noise and improve the selection of measured marginals for learning the distribution. By selecting the marginal costing unit privacy budget that offers the highest benefit, we can prioritize the selection process and improve the overall accuracy of the synthetic data generation, especially when the privacy budget is limited at a low level.


\section{Preliminaries}
\subsection{Marginals}

A marginal refers to a low-dimensional statistic computed from 
a high-dimensional data. It is essentially a table that counts the occurrences of each possible combination of values within a specific set of attributes. In other words, it provides a summary of the data by focusing on a subset of attributes and their corresponding value combinations. Marginals play a crucial role in various data analysis tasks, as they capture important statistical information about the relationships and distributions of the underlying high-dimensional data. 

\begin{definition}
(Marginal ~\cite{McKenna2019GraphicalmodelBE}). Let $r\subseteq[d]$ be a subset of attributes, $\boldsymbol{\omega_r}=\prod_{i\in r}\boldsymbol{\omega}_{i}$ and $x_{r}=(x_{i})_{i\in r}$. The marginal on $r$ is a vector $\mu \in \mathcal{R}^{n_{r}}$, indexed by domain elements $t\in \boldsymbol{\omega}_{r}$, such that each entry is a count, i.e., $\mu [t]=\sum_{x\in D}\mathbb{1}[x_{r}= t]$. We let $M_{r}: \mathcal{D}\to \mathcal{R}^{n_{r}}$ denote the function that computes the marginal on $r$, i.e., $\mu =M_{r}(D)$.
\end{definition}

We give an example of marginal as follows. 
\begin{example}
  Suppose we have a dataset with three attributes: $A$, $B$, and $C$. Attribute $A$ have two values $a_1$, $a_2$, while attribute $B$ have two values $b_1$ and $b_2$. Table $Q_i$ in Fig~\ref{Fig:framework} represents the marginal on attributes ($A$, $b$) which contains the counts for all the possible combination of values for attributes $A$ and $B$.
\end{example}

As the dimensionality of the data increases, computing all the marginals becomes computationally expensive. Therefore, we need to identify a subset of marginals that effectively capture and represent the high-dimensional data.

\subsection{Differential Privacy}
Differential privacy was originally proposed by Dwork in 2006 for the privacy leakage problem of statistical databases, which ensures that any individual in the dataset or not has little effect on the final released query results.

\begin{definition}
(Differential Privacy). A randomized mechanism $\mathcal{M}: \mathcal{D} \to \mathcal{R}$ satisfies $(\epsilon,\delta)$-differential privacy (DP) if for any neighboring datasets $D,D'\in \mathcal{D}$, $dis(D,D')=1$, and any subset of possible outputs $S\subseteq \mathcal{R}$,
\begin{equation}\label{Def:LDP}
 Pr[\mathcal{M}(D)\in S]\leq exp(\epsilon)\cdot Pr[\mathcal{M}(D')\in S]+\delta
\end{equation}
\end{definition}

Sensitivity is a key metric that determines the privacy of a mechanism, defined below.

\begin{definition}
  (Sensitivity). For $f: \mathcal{D}\to \mathbb{R}^{h}$, and all the neighboring datasets $D,D'\in \mathcal{D}$ ($dis(D,D')=1$), the $L_{2}$ sensitivity of $f$ is 
  \begin{equation}
   \Delta(f)=max_{D,D'}||f(D)-f(D')||_{2}.
  \end{equation}
\end{definition}

The sensitivity of a marginal query is commonly $1$, because one individual can only contribute a count of one to a single cell of the marginal. 

In this paper, we adopt the \textbf{Gaussian Mechanism} to add noise to the chosen margial, and use the \textbf{Exponential Mechanism} for marginal selection. We use the \textbf{zero-Concentrated Differential Privacy (zCDP)} to provide a composition of different DP mechanism. 

\begin{definition}
(Gaussian Mechanism ~\cite{DBLP:conf/tcc/BunS16}). For $f: \mathcal{D}\to \mathbb{R}^{h}$. The Gaussian Mechanism adds i.i.d. Gaussian noise with scale $\sigma \Delta(f)$ to each entry of $f(D)$.
\begin{equation}
  \mathcal{M}(D)=f(D)+\sigma \Delta (f)\mathcal{N}(0,\mathbb{I})
\end{equation}
where $\mathcal{I}$ is a $h\times h$ identity matrix.
\end{definition}

\begin{definition}
  (Exponential Mechanism). Consider a query q(x):$\mathcal{D}\to \mathbb{R}$ and specifie a scoring function $u :\mathcal{D} \times \mathcal{R} \to \mathbb{R} $, when answering it, the exponential mechanism selects a response from a finite set based on the scoring function, and the response satisfies the following probability:
\begin{equation}
  Pr[\mathcal{M}(D)=r] \propto exp \left(\frac{\epsilon}{2\Delta}\cdot q(D) \right)
\end{equation}
where $r\in \mathcal{R}$ and $\Delta$ is global sensitivity.
\end{definition}

\begin{definition}
(zero-Concentrated Differential Privacy (zCDP)). A randomized mechanism $\mathcal{M}$ is $\rho\text{-}zCDP$ if for any two neighboring datasets $D$ and $D'$, and all $\alpha \in (1,\infty)$, we have:
\begin{equation}
  D_{\alpha}(\mathcal{M}||\mathcal{M}(D'))\le \rho \cdot \alpha
\end{equation}
where $D_{\alpha}$ is the Rényi divergence of order $\alpha$.
\end{definition}

\begin{theorem}\label{Composition_zCDP}
  (Composition of zCDP). Two randomized mechanisms $\mathcal{M}_1$ and $\mathcal{M}_2$ satisfy  $\rho_1\text{-}zCDP$ and $\rho_2\text{-}zCDP$, respectively, their sequential composition $\mathcal{M}=(\mathcal{M}_1,\mathcal{M}_1)$ satisfies $(\rho_1+\rho_2)\text{-}zCDP$.
\end{theorem}

\begin{proposition}
  (zCDP of Gaussian Mechanism). The Gaussian Mechanism with scale $\sigma$ satisfies $\frac{1}{2\sigma^2}\text{-}zCDP$.
\end{proposition}

\begin{proposition}
  (zCDP of Exponential Mechanism). The Exponential Mechanism with privacy budget $\epsilon$ satisfies $\frac{\epsilon^2}{8}\text{-}zCDP$.
\end{proposition}

By leveraging these two propositions, we can convert the privacy budget of mechanisms such as the Gaussian Mechanism and the Exponential Mechanism into that of the zCDP. This allows us to effectively manage the composition of privacy guarantees using the theorem~\ref{Composition_zCDP}.

\subsection{Private-PGM}
We just use the private-PGM as a black box to approximate the jointly distribution $Pr[A_1, A_2, ..., A_n]$ via constructing junction tree from a list of noisy marginals $\tilde{W}=\{\tilde{w_1}, \tilde{w_2},... ,\tilde{w_n}\}$.


The core of private-PGM can be seen as solving an optimization problem:
\begin{equation}
  \hat{p}\in \mathop{\arg\min}\limits_{p\in S}\sum\limits_{i=1}^{k}\frac{1}{\sigma_{i}} ||M_{w_{i}}(p)-\tilde{w_{i}}||_{2}^{2}
\end{equation}
, where $\mathcal{S}=\{p|p(x)\ge 0 \text{ and } \sum_{x\in \mathcal{X}}p(x)=n\}$ is the set of probability distributions over the full domian, and $\tilde{w_{i}}$ is the noisy result of marginal $w_i$.



\section{Partition-based private data synthesis}

In order to accurately estimate a distribution that closely resembles the true distribution based on noisy marginals, it is crucial to carefully select the marginals that provide significant value while requiring minimal privacy budget. 
We have developed a new method called PPSyn, which is a partition-based private data synthesis approach. Its primary purpose is to decrease the privacy budget needed for each marginal, based on partition techniques. Furthermore we have developed a strategy to allocate the privacy budget based on the positive impact of enhancing the accuracy of estimated distribution per unit budget for each selected marginal. This allows us to make full use of the limited privacy budget, ensuring that it is allocated in the most effective manner.

In Section~\ref{Sec:FrameworkOfSyn}, we present the framework of our proposed method PPSyn.   In Section~\ref{Sec:Partition}, we introduce two partition methods specifically designed for one-dimensional and multidimensional marginals, respectively. 

\subsection{The framework of partition-based private data synthesis (PPSyn)}\label{Sec:FrameworkOfSyn}

Given the limited privacy budget, our method focuses on maximizing information the distribution model learned while minimizing privacy expenditure. To achieve this objective, we employ two steps.
First, we propose partition methods to decrease the number of cells in each marginal. By reducing the number of cells, we effectively reduce the amount of DP noise. Secondly, we prioritize the selection of marginals that provide the highest contribution to improve the model per unit  budget. This ensures that our method maximizes the information gained while making efficient use of the limited privacy budget available.

The pseudo-code for our framework is depicted in Algorithm~\ref{Algorithm:PrivSynLB}. This algorithm allocates $0.1\rho$ of the privacy budget to learn the distribution of multidimensional data $\hat{p}$ using noisy one-way marginals (line 1). The privacy budget allocated for each marginal is determined by PrivSyn~\cite{Zhang2020PrivSynDP}. Subsequently, the algorithm trains the PrivatePGM model based on the noisy one-way marginals (line 2-3). We also allocate $0.1\rho$ for marginal selection. we adopt exponential mechanism to protect the privacy for selectiong a new measurements from the workload, and the privacy budget for each iteration is $\rho_{exp}=0.1\rho/T$ (line 4). 

\begin{algorithm}
  \caption{Partition-based private data synthesis(PPSyn)}\label{Algorithm:PrivSynLB}
  {\bf Input:} Private Dataset $D$, Marginals $W$, Privacy budget $\rho$\\
  {\bf Output:} Synthetic dataset  $\hat{D}$\\
  {\bf Hyper-Parameters:} Number of rounds $T$
  \begin{algorithmic}[1]
  \State Initial the distribution $\hat{p}$ by training the PrivatePGM according to one-way marginals with 0.1$\rho$, i.e., $\rho_{used}\leftarrow0.1\rho$.
   The privacy budget for a one-way marginal $w_i$ in $W$ is $\rho_i=\frac{c_i^{2/3}}{\sum_{j\in{one}}c_j^{2/3}}\cdot0.1\rho$.  
  \State $measures\leftarrow$ noisy one-way marginals
  \State $\hat{p}\leftarrow PrivatePGM(measures)$
  \State $\rho_{exp}\leftarrow \frac{0.1\rho}{T}$
  \While {$\rho_{used}<\rho$}
      \State select: $m, P_m,\rho_m\leftarrow PartSele(W,\hat{p},\rho-\rho_{used},\rho_{exp})$ 
      \State add noise: $\tilde{P_m}\leftarrow \tilde{P_m}+ Noise(\rho_m)$ 
      \State measure: $measures\leftarrow measures\cup Est(m, \tilde{P_m})$
      \State estimate: $\hat{p}\leftarrow PrivatePGM(measures)$ 
      \State update: $\rho_{used}\leftarrow\rho_{used}+\rho_m+\rho_{exp}$
      \EndWhile
  \State generate synthetic data $\hat{D}$ with $\hat{p}$
  \State \textbf{return: $\hat{D}$}
  \end{algorithmic}
\end{algorithm}
In an iterative manner, the algorithm selects one high-quality marginal and include the noisy marginal in the measurements to update the PrivatePGM model (line 5-10).
In each iteration, the algorithm employs partition-based marginal selection methods (to be discussed in section~\ref{Sec:Partition} with more details) to identify an appropriate partition for a specific marginal from the workload. It reduces the number of cells of each marginal using partition methods and returns the selected marginal $m$ along with its corresponding partition $P_m$ and the required budget $\rho_m$. The required budget makes sure that the selected marginal brings sufficient contribution to improve the model.
It then adds noise to the selected partition using the Gaussian Mechanism, providing estimations  $Est(m, \tilde{P_m})$  for each cell of the selected marginal $m$ based on the noisy partition $\tilde{P_m}$ under the assumption of uniform distribution. Finally, the algorithm generates synthetic data according to the estimated distribution (line 12).


\subsection{Partition-based Marginal Selection}\label{Sec:Partition}


In marginal-based synthetic data generation, multiple marginals are utilized as measurements to estimate the distribution of synthetic data. Each marginal consumes a certain privacy budget, and the amount of noise introduced according to DP is inversely related to the budget but directly related to the size of the marginal. Therefore, in order to capture a noisy marginal with high utility, a larger budget needs to be allocated for larger marginals. However, due to the limited total budget, allocating more to large marginals would result in reducing the number of measurements available, potentially leading to insufficient data to accurately approximating the true distribution. On the other hand, excluding large marginals may result in the loss of important distribution information.

We observe that most of the data is concentrated in a small region, meaning that a large number of cells in large marginals have small values that are not crucial for estimating the distribution but contribute to introducing unnecessary noise. To address this, we propose partition-based marginal selection methods to reduce the size of each marginal based on a partition. This allows us to add noise to each interval of the partition instead of each cell of the marginal, effectively reducing the privacy budget required while maintaining similar utility. With this approach, we can capture important information from large marginals without incurring excessive privacy budget expenditure.

To better understand the impact of noise, we provide some fundamental definitions that describe the error associated with the noisy marginals.

\begin{definition}
 \textbf{Noisy Marginal Error.} We consider the distribution of a marginal over $n$ elements. The noisy marginal error is the l1-norm distance between the noisy marginal and the true marginal $\left\|\tilde{p}-q\right\|_1$, where $q(i)$ is the true count for element $i$ of the marginal and $\tilde{p}(i)$ is the noisy count after adding noise according to differential privacy methods such as Gaussian Mechanism.
\end{definition}

\begin{definition}\label{Def:NoiseMarginalWithoutPartition}
 \textbf{Noisy Marginal without Partition.} Given a marginal $w$ whose size is $n$, the noisy marginal without partition is the result of adding noise to each cell in the marginal according to Gaussian Mechanism. The estimation for a noisy marginal without partition is $\tilde{p}$, and $\tilde{p}(i)=q(i)+\mathcal{N}(0,\sigma^2)$.
\end{definition}

We use the following theorem to show the impact of adding noise to a marginal.

\begin{theorem}\label{Theo:NoiseMarginalWithoutPartition}
  Given a marginal $w$ whose size is $n$, the l1-error of the noisy distribution $\tilde{p}$ from the true distribution $q$ is $\mathbb{E}[\left\|\tilde{p}-q\right\|_1]=\sqrt{2/\pi}n\sigma$, where $n$ is the number of cells of $w$.
\end{theorem}
\begin{proof}
  Since $error(i)=\tilde{p}(i)-q(i)$, $error(i)\sim\mathcal{N}(0,\sigma^2)$, the expectation error for the $i$th cell of the marginal is $\mathbb{E}[|\tilde{p}(i)-q(i)|]=\sqrt{2/\pi}\sigma$. Therefore, $\mathbb{E}[\left\|\tilde{p}-q\right\|_1]=\sum_{i\in[n]}|\tilde{p}(i)-q(i)|=\sqrt{2/\pi}n\sigma$.
\end{proof}

According to the above theorem, the noise error depends on the size of each marginal $n$ (the number of cells of each marginal). Adding noise to each cell in the marginal directly protect the privacy but can lead to significant errors when dealing with a marginal including a large number of elements. In order to reduce the noise error, partition is a strategy. If we divide the data into intervals, we can add noise to each interval instead of each element and estimate the elements within based on the assumption of uniform distribution. Partitioning helps reduce noise error but introduces reconstruction error. We use the follow definition to show the impact of noise on the marginal with partition.

\begin{definition}\label{Def:NoiseMarginalWithPartition}
 \textbf{Noisy Marginal with Partition.} The noisy marginal with partition separates the marginal into several disjoint intervals $P=\{I_1, I_2, ..., I_l\}$, and the count for each interval $I$ is the sum of all the element counts in the interval $q(I)=\sum_{i\in I}q(i)$. Gaussian noise is added to each $p(I)$. The estimation for a noisy marginal with partition is $\tilde{p}$, and $\tilde{p}_{i\in I}(i)=\frac{q(I)+\mathcal{N}(0,\sigma^2)}{|I|}$.
\end{definition}

\begin{theorem}\label{Theo:NoiseMarginalWithPartition}
  Given a marginal $w$ whose size is $n$, the l1-error of the estimated distribution $\tilde{p}$ based on the noisy partition $\tilde{P}$ is $\mathbb{E}[\left\|\tilde{p}-q\right\|_1]\le\left\|q(P)-q\right\|_1+\sqrt{2/\pi}\cdot|P|\cdot\sigma$, where $|P|$ is the number of intervals of $P$.
\end{theorem}
\begin{proof}
   For each cell $i$ in the original marginal, $error(i)=\tilde{p}(i)-q(i)$, where $\tilde{p}(i)$ is estimated value for cell $i$ according to the noisy partition. 
   $\mathbb{E}[\left\|\tilde{p}-q\right\|_1]
   =\sum_{i\in [n]}|\tilde{p}(i)-q(i)|
   \le \sum_{I_j\in P}\sum_{i\in I_j}|\frac{q(I_j)}{|I_j|}-q(i)|+\sum_{I_j\in P}|\tilde{p}(I_j)-q(I_j)|$.

   According to theorem~\ref{Theo:NoiseMarginalWithoutPartition}, $E[\sum_{I_j\in P}|\tilde{p}(I_j)-q(I_j)|]=\sqrt{2/\pi}\cdot|P|\cdot\sigma$. The reconstruction error is $\sum_{I_j\in P}\sum_{i\in I_j}|\frac{q(I_j)}{|I_j|}-q(i)|=\left\|q(P)-q\right\|_1$, therefore, $\mathbb{E}[\left\|\tilde{p}-q\right\|_1]\le\left\|q(P)-q\right\|_1+\sqrt{2/\pi}\cdot|P|\cdot\sigma$.
\end{proof}

Definitions~\ref{Def:NoiseMarginalWithoutPartition} and \ref{Def:NoiseMarginalWithPartition} show the difference of protecting the privacy of the marginal with and without partition, respectively. Adding noise to the count of each element in the marginal directly introduce too much noise error. Partitioning reduces noise error but introduces reconstruction error. We can learn from theorems~\ref{Theo:NoiseMarginalWithoutPartition} and ~\ref{Theo:NoiseMarginalWithPartition} that the partition can reduce the total error as long as $\left\|q(P)-q\right\|_1<\sqrt{2/\pi}\cdot(n-|P|)\cdot\sigma$, i.e., the reduction in noise error outweighs the increase in reconstruction error.

\subsubsection{The framework of marginal selection}
In this section, we present the partition-based marginal selection method that we propose. Our objective is to identify the marginal that contributes the most while utilizing the privacy budget efficiently. The error of each marginal is influenced by two factors: the noise error introduced by DP and the reconstruction error introduced by the partitioning process. Instead of pre-determining the privacy budget for each marginal, we dynamically allocate the privacy budget based on its corresponding partition. Finally, we select the measurement with the highest score in terms of the contribution while costing per unit budget.

\begin{algorithm}
\caption{Partition-based Marginal Selection (PartSele)}\label{Algorithm:Sele}
{\bf Input:} Marginals $W$, distribution $\hat{p}$, budget for exponential mechanism $\rho_{exp}$\\
{\bf Output:} Selected marginal $m$, Partition $P_m$, required budget $\rho_m$\\
{\bf Hyper-Parameters}: Error impact $\eta$
\begin{algorithmic}[1]
\For {each $w$ in $W$}
    \State $Contr_w\leftarrow \left\|M_{w}(D)-M_{w}(\hat{p}) \right\|_1$
    \State $P_w,\rho_w \leftarrow Partition(w,Contr_w,\eta)$
    \State $score[w]=\frac{Contr_w}{\rho_w}$
\EndFor
\State $m\leftarrow ExponentialMechanism(W,score,\rho_{exp})$
\State \textbf{return: $m$, $P_m$, $\rho_m$}
\end{algorithmic}
\end{algorithm}

The pseudo-code of the partition-based marginal selection method is shown in Algorithm~\ref{Algorithm:Sele}. In each iteration, we compute the $l_1$-error between the true marginal $w$ on the dataset $D$ and the marginal estimated according to the current distribution $\hat{p}$ (line 2). We regard this error as the potential contribution of each marginal for the improvement of the estimated distribution $\hat{p}$.
Then we compute a proper partition $P_w$ along with the required privacy budget $\rho_w$  according to the marginal $w$ and the contribution $Contr_w$ (line 3). The parameter $\eta$ regards the impact of noise on the contribution of this marginal. We define a score function for each marginal to measure its value for the distribution, and the score is positively correlated with its contribution brought by each unit of budget expenditure(line 4). After computing the score for each marginal, the algorithm selects one marginal according to the \textit{ExponentialMechanism} (line 6).

\subsubsection{Marginal partition}
We propose two partition methods for one-dimensional and multi-dimensional marginal, respectively. The partition methods aim to reduce noise error by merging some uniformly distributed cells into one larger interval. Partitioning reduces noise error but introduces reconstruction error. In order to guarantee the positive contribution that each candidate marginal brings, we limit the impact of total noise to the contribution of one marginal and find a proper partition that satisfies this constraint for each marginal.

\textbf{\underline{One-dim marginal partition}}. For each one-dim marginal, we begin by treating each cell of the true marginal as an interval. We then iteratively merge two disjoint intervals that result in the minimum reconstruction error.

\begin{algorithm}
  \caption{One-dim marginal partition}\label{Algorithm:1D-Partition}
  {\bf Input:} 1-dim marginal $w$, Contribution $Contr$, Noise impact $\eta$\\
  {\bf Output:} Partition $P$, privacy budget $\rho$.
  \begin{algorithmic}[1]
  \State Sort the cells in $w$ in ascending order of the counts $\{(i_1,c_1),...,(i_n,c_n)\}$, $n$ is the number of cells in $w$.
  \State Initial the partition $P_0$ and regarding each cell as an interval.
  \State $t=0$
  \State $\rho_t\leftarrow \frac{|P_t|^2}{\pi\cdot(\eta\cdot Contr)^2}$
  \While {1}
      \State $c \leftarrow \mathop{\arg\min}\limits_{u\in\{1,2,...,|P_t|-1\}} MergeError(I_{u},I_{u+1})$
      \State $P_{t+1} \leftarrow$  merge the intervals $I_c$ and $I_{c+1}$ in $P_{t}$.
      \State $RE\leftarrow \left\|M_{w}(D)-M_{w}(P_{t+1}) \right\|_1$ 
      \State $\rho_{t+1}\leftarrow \frac{|P_{t+1}|^2}{\pi\cdot(\eta\cdot Contr-RE)^2}$
      \If {$\rho_{t+1}>\rho_{t}$}
          \State break;
      \EndIf
      \State $t\leftarrow t+1$
  \EndWhile
  \State \textbf{return: $P_t$, $\rho_t$}
  \end{algorithmic}
  \end{algorithm}

The pseudo-code for the one-dimensional marginal partition method is shown in Algorithm~\ref{Algorithm:1D-Partition}. Given a one-way marginal $w$, its potential contribution $Contr$, and the error impact parameter $\eta$, the algorithm begins by treating each cell in the marginal as an interval within the partition (line 1-2). The initial required privacy budget $\rho_0$ is computed based on $\eta\cdot Contr$ and $|P_0|=n$ (line 4). This setting ensures that the total error, including both the noise error and the reconstruction error, does not exceed $\eta\cdot Contr$. The algorithm then continues merging interval pairs that result in the minimum reconstruction error (line 6), gradually reducing the size of the partition (line 7). At each iteration $t$, the algorithm updates the reconstruction $RE$ (line 8) and the required privacy budget $rho_t$ for the current partition $P_t$ (line 9). We will use an example to demonstrate that the required privacy budget initially decreases as the size of the partition  $|P_t|$ increases, but then starts to increase again. Since our goal is to find the minimum required privacy budget while maintainting the its contribution to improve the distribution approximation, the algorithm terminates when the required budget begins to increase (line 10-12).


To explain the reason why the algorithm terminates when the budget begins to increase. We illustrate the impact of partition on the required privacy budget of each marginal using the following example.

\begin{figure}[htbp]
  \centering
  \includegraphics[scale=0.5]{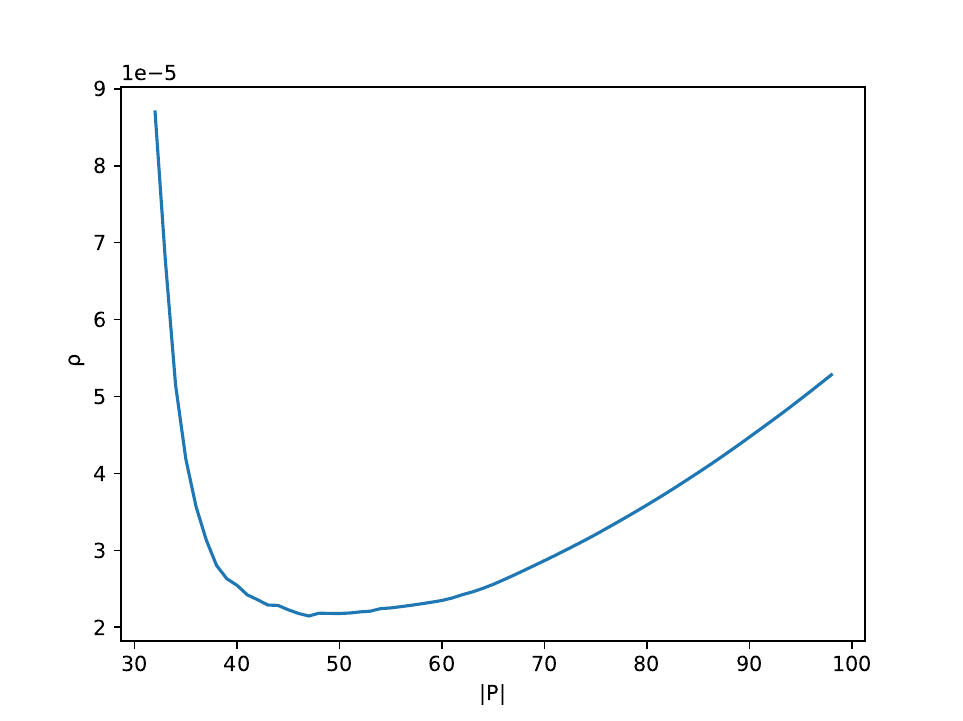}
  \caption{Impact of $|P|$ on the required privacy budget $\rho$}
  \label{Fig:example}
  \end{figure}

\begin{example}
Let's consider a one-way marginal on the attribute ``hours-per-week'' that includes 99 cells. As shown in Figure~\ref{Fig:example}, the required privacy budget ($\rho$) decreases as the size of the partition ($|P|$) decreases. However, once it reaches the minimum value, it starts to increase again. Therefore, by treating each cell as an interval within the partition and continuously merging the intervals that result in the minimum reconstruction error, we can find the optimal value that minimizes the required privacy budget.
\end{example}

We use the following theorem to prove that the error of the marginal estimated according to the noisy partition returned by Algorithm~\ref{Algorithm:1D-Partition} is limited.

\begin{theorem}\label{theorem:errorlimit}
  Suppose $P$ and $\rho$ is the partition and privacy budget for a marginal $w$ returned by algorithm~\ref{Algorithm:1D-Partition}. Thus,
   $\left\|M_{w}(D)-M_{w}(\tilde{P}) \right\|_1\leq \eta\cdot Contr$, where $M_{w}(D)$ is the true marginal on data $D$, and $M_{w}(\tilde{P})$ is the marginal estimated according to the noisy partition $P$.
\end{theorem}

\begin{proof}
According to triangle inequality,
\begin{equation}
   \left\|M_{w}(D)-M_{w}(\tilde{P}) \right\|_1 \leq \left\|M_{w}(D)-M_{w}(P) \right\|_1 
   +\left\|M_{w}(P) -M_{w}(\tilde{P}) \right\|_1
\end{equation}
where $\left\|M_{w}(D)-M_{w}(P) \right\|_1$ is the reconstruction error $RE$ caused by partition, and $\left\|M_{w}(P) -M_{w}(\tilde{P}) \right\|_1$ is the noisy error $NE$ caused by Gaussian Mechanism.
Algorithm ~\ref{Algorithm:1D-Partition} computes the privacy budget as  $\rho =\frac{|P|^2}{\pi\cdot(\eta\cdot Contr-RE)^2}$, thus, the noise error can be computed as:
\begin{equation}
  NE=\frac{|P|}{\sqrt{\pi\cdot \rho}}=\eta\cdot Contr-RE
\end{equation}
Therefore, $\left\|M_{w}(D)-M_{w}(\tilde{P}) \right\|_1 \leq RE+NE\leq \eta\cdot Contr$. Since $\eta<1$, the algorithm makes sure that each chosen marginal brings sufficient positive contribution to learning the distribution.
\end{proof}

We can learn from the above theorem that the impact of noise on the contribution each marginal brought is limited.

%
%
%
%
%
We also extend the idea of using partition to reduce the noise error to multidimensional marginals. The size of a multidimensional marginal is typically large, as it is determined by the product of the domain sizes of each dimension. Reducing the size of the multidimensional marginal is significant in order to minimize the noise error, as the noise error is directly related to the size of each marginal.

\textbf{\underline{Multidimensional marginal partition}}. Partitioning for a multidimensional marginal is more of a trial compared to partitioning for a one-dimensional marginal. This is because determining the optimal partition for even two-dimensional data is an NP-complete problem. It becomes difficult to iteratively merge two multidimensional interval pairs into one. Therefore, a common approach is to initially consider the entire domain as a multidimensional interval and then iteratively split the interval with the highest reconstruction error into two, as this is a more feasible strategy.

\begin{algorithm}
\caption{Multi-dim marginal partition}\label{Algorithm:MD-Partition}
{\bf Input:} Multi-dim marginal $w$, Contribution $Contr$, Noise impact $\eta$\\
{\bf Output:} Partition $P$, privacy budget $\rho$.
\begin{algorithmic}[1]
\State Initial the partition $P_0$ by regarding the whole domain as an multi-dimensional interval.
\State $d\leftarrow$ dimension of $w$
\State $t=0$
\State $RE\leftarrow \left\|M_{w}(D)-M_{w}(P_{t}) \right\|_1$
\If {$RE<\eta\cdot Contr$}
    \State $\rho_t\leftarrow \frac{|P_t|^2}{\pi\cdot(\eta\cdot Contr-RE)^2}$
\Else
    \State $\rho_t\leftarrow +\infty$
\EndIf
\While {1}
    \State $d_c\leftarrow t\% d$
    \State $SP\leftarrow$ find the split point in the  $d_c$th dimension with maximum RE reduction
    \State $P_{t+1} \leftarrow$  split $P_{t}$ with $SP$.
    \State $RE\leftarrow \left\|M_{w}(D)-M_{w}(P_{t+1}) \right\|_1$ 
    \If {$RE<\eta\cdot Contr$}
        \State $\rho_{t+1}\leftarrow \frac{|P_{t+1}|^2}{\pi\cdot(\eta\cdot Contr-RE)^2}$
        \If {$\rho_{t+1}>\rho_{t}$}
            \State break;
        \EndIf
    \EndIf
    \State $t\leftarrow t+1$
\EndWhile
\State \textbf{return: $P_t$, $\rho_t$}
\end{algorithmic}
\end{algorithm}


The pseudo-code for the multi-dimensional marginal partition method is shown in Algorithm~\ref{Algorithm:MD-Partition}. Due to the complexity of iteratively merging multi-dimensional interval pairs, the algorithm begins by considering the entire domain as a multi-dimensional interval (line 1). The variable $d$ represents the dimensionality of the marginal (line 2). The algorithm initializes the reconstruction error RE for the current partition (line 4). If RE is less than $\eta\cdot Contr$, it computes the required privacy budget $\rho_t$; otherwise, $\rho_t$ is initialized as $+\infty$.
The algorithm iteratively selects a split point in different dimensions using a round-robin strategy (line 11), and splits the interval into two parts in a way that maximizes the reduction in reconstruction error (line 12). Since the algorithm starts with a single interval, the reconstruction error in the initial iterations may exceed the threshold $\eta\cdot Contr$. To ensure that the total error, including both the reconstruction error and the noise error, is smaller than $\eta\cdot Contr$, we begin checking the required privacy budget after RE falls below $\eta\cdot Contr$ (line 15-20). The algorithm also aims to find the minimum required privacy (line 17-19), which is the same as in Algorithm~\ref{Algorithm:1D-Partition}.

\begin{theorem}
  Suppose $P$ and $\rho$ is the partition and privacy budget for a marginal $w$ returned by algorithm~\ref{Algorithm:MD-Partition}. Thus,
   $\left\|M_{w}(D)-M_{w}(\tilde{P}) \right\|_1\leq \eta\cdot Contr$, where $M_{w}(D)$ is the true marginal on data $D$, and $M_{w}(\tilde{P})$ is the marginal estimated according to the noisy partition $P$.
\end{theorem}
\begin{proof}
  This theorem can be proven using the same method as the proof of Theorem~\ref{theorem:errorlimit}.
\end{proof}

According to the partition algorithms we have proposed for both one-dim and multi-dim marginals, we are able to reduce the size of large marginals and subsequently decrease the privacy budget needed for them. This provides us with the potential to utilize more marginals in order to enhance the quality of the distribution model.

\textbf{\underline{Skill to improve efficiency}}. In Algorithm~\ref{Algorithm:Sele}, computing a new partition for each marginal in every iteration of marginal selection may seem costly. However, we can improve the efficiency by precomputing and storing the partitions  $|P_0|$, $|P_1|$, ..., $|P_n|$ in advance. Since both the interval pair to be merged for one-dimensional marginal and the split points for multi-dimensional marginal are determined by the reconstruction error, we can avoid repeatedly computing the partition in later iterations. Instead, the algorithm only needs to find the termination point $|P_t|$ determined by the potential contribution in each iteration. This optimization significantly reduces the computational overhead and enhances the efficiency of the algorithm.


\section{Experiments}

In this section we will empirically evaluate PPSyn, comparing it against various state-of-the-art mechanisms, while also analyzing the impact of different parameters within the methodology.

\subsection{Experimental Setup}

In this section, we will introduce the experimental setup of our method, including the source of the utilized datasets. We will also provide error metrics and conduct experiments on two crucial parameters that influence the performance of the method.

\noindent \underline{\textbf{Hardware}}

All our experiments are implemented in Python 3.8 and run on Intel Xeon Gold 5218R@2.10GHz and 32GB memory.

\noindent \underline{\textbf{Datasets}}

We evaluate the performance of our method on three datasets, including three real-world datasets.

(1) \textbf{ADULT}\footnote{https://archive.ics.uci.edu/dataset/2/adult}. The first dataset is a real-world dataset consisting of 48, 842 records from the 1996 US Census data. It contains personal information, including 14 attributes such as age, gender, race, salary, etc. The entire domain size is approximately 65,500.

(2) \textbf{BR2000}\footnote{https://usa.ipums.org/usa/cite.shtml}. The second dataset contains 38,000 records of personal information obtained from Brazil Census. It features 14 attributes, such as gender, work address, and more. 

(3) \textbf{LOANS}\footnote{https://github.com/giusevtr/fem/tree/master/datasets}. The last one dataset is coming from the the UCI repository. It consists of 42,535 records whose features include 48 attributes. It is also the dataset used by methods like MST. It is also used in FEM ~\cite{Vietri2020NewOA} and RAP ~\cite{DBLP:conf/icml/AydoreBKKM0S21}.

\noindent \underline{\textbf{Competitors}}

(1) AIM~\cite{McKenna2022AIMAA}  is the state-of-the-art private synthetic data generation method, it  is based on workload and aims to infer the joint probability distribution of the entire original database using the query set provided by the user. It follows the select-measure-generate paradigm, which is a common approach used in similar methods. Its uniqueness lies in AIM proposing a novel scoring function, which  in each iteration, enables the selection of marginals that are more conducive to PGM learning compared to other similar methods.

(2) MST~\cite{McKenna2021WinningTN} is also a workload-based method, does not depend on the availability of public or provisional data but uses a portion of the privacy budget to measuring marginals to  obtain a better initial distribution for probabilistic graphs. Its distinct feature is the modularization of the select-measure-generate three components, making it akin to a baseline method.

(3) PrivMRF~\cite{Cai2021DataSV} differs from the above two methods as it is a data-based heuristic approach that considers the correlations among all selected attributes while ensuring $\theta$-useful (A criterion that ensures the selected marginals under a limited privacy budget are low-dimensional and capable of capturing essential features of the input data.). Unlike the previous two approaches, PrivMRF departs from the entire dataset to select marginals that satisfy $\theta$-useful criteria for participation in the training of the probabilistic graph. 

\noindent \underline{\textbf{Error Metrics}}

In this paper, we consider the following three error metrics. In the following metrics, $W$ is the workload, $Q$ is the set of range queries,  $n$ denotes the number of queries, $p$ represents the true marginal, and $\hat{p}$ represents the marginal estimated based on the synthetic dataset.

(1) Workload error. The mean squared error (MSE) of marginals in the workload 
$W$:$\frac{1}{n\cdot|W|} \sum_{w\in W} \left\|p(w)- \hat{p}(w)\right\|_1$.

(2) Range query error. The mean squared error (MSE) of range queries in $Q$: $\frac{1}{n} \sum_{q\in Q} (p(q)- \hat{p}(q))^{2}$.

(3) Classification error. The mis-classification rate of the SVM model when trained on synthetic data generated by different methods: $1-(TP+TN)/N$. $TP$ stands for the number of samples correctly predicted as positive, $TN$ stands for the number of samples correctly predicted as negative, and $N$ represents the total number of samples. 



\noindent \underline{\textbf{Parameters}}


$\epsilon$: The privacy budget.

$d$: The dimensions of the workload significantly impacts the performance of workload-based methods, as it dictates which marginals should be inferred from the probabilistic graph to approximate the distribution.

$\eta\in (0,1)$: It defines the maximum permissible level of error during the marginal selection phase. By setting this constraint, the algorithm can prioritize the selection of marginals with lower total error, thereby increasing the chances of selecting marginals that are more valuable for the distribution.

\begin{figure*}
  \centering
    \subfigure[ADULT.]{
      \label{Fig:ADULT_W}
      \includegraphics[width=0.3\linewidth]{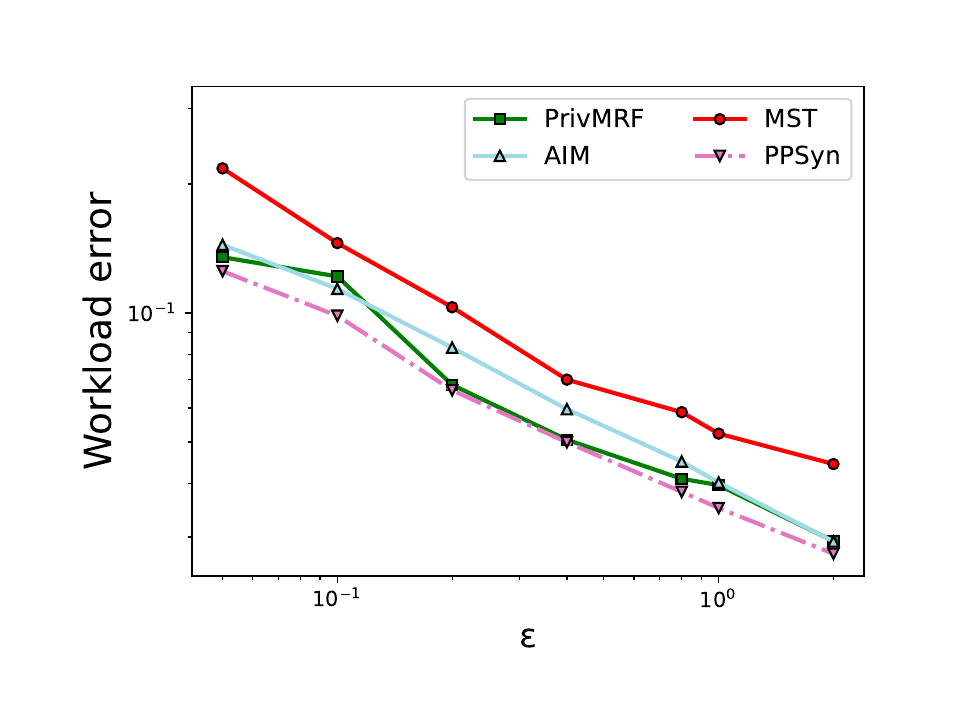}}
    \subfigure[BR2000.]{
      \label{Fig:BR2000_W}
      \includegraphics[width=0.3\linewidth]{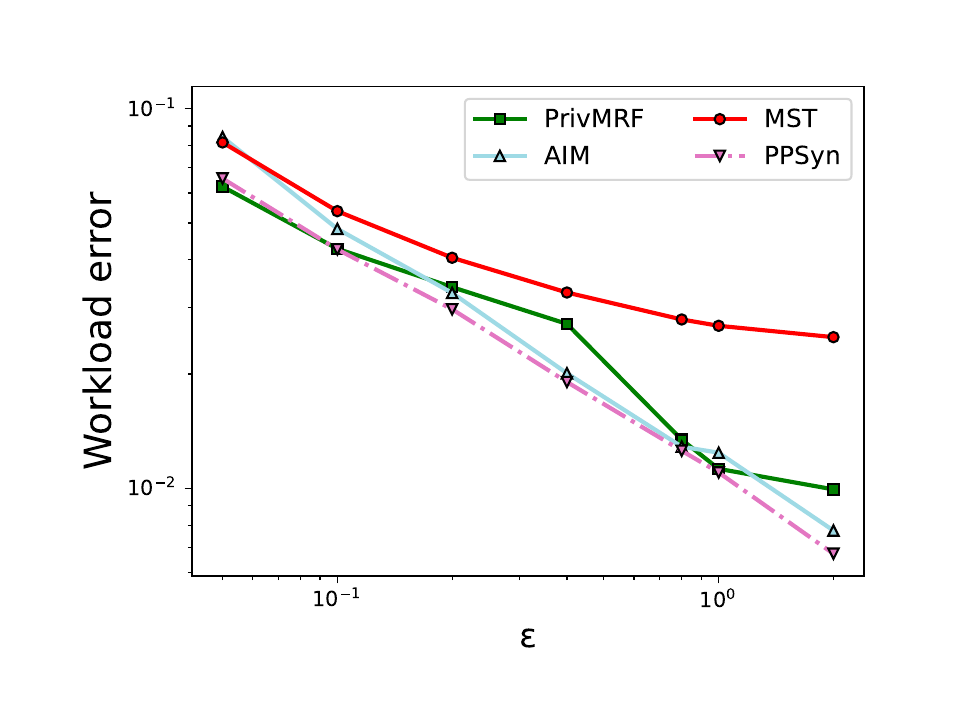}}
    \subfigure[LOANS.]{
      \label{Fig:LOAN_W}
      \includegraphics[width=0.3\linewidth]{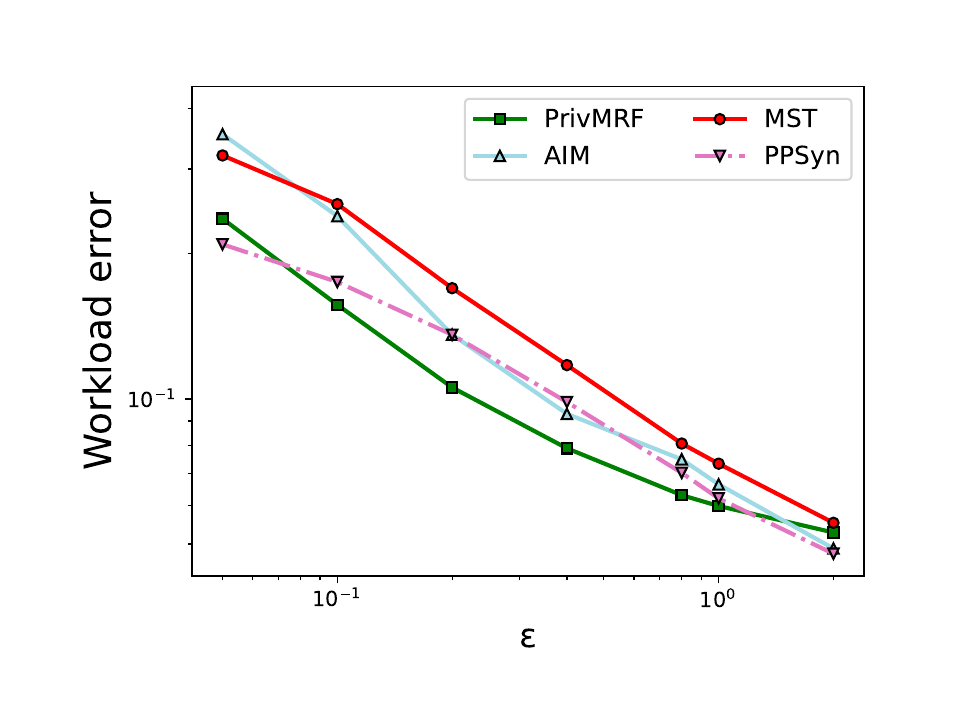}}
  \caption{The Accuracy of Worklaod}
  \label{Fig:accuracy_Workload}
  \end{figure*}

In this section, we will empirically evaluate the accuracy of PPSyn using three error metrics. ``workload error'' reflects the average error between synthetic data and original data, while ``range query error'' assesses the synthetic data's ability to respond to range queries. We also use the ``classification error'' to test whether the synthetic data retains the expected characteristics for each attribute.

\subsubsection{Workload error}
We test the workload error of different methods on three datasets, and the results are shown in Figure~\ref{Fig:accuracy_Workload}. 
The workload is derived from random selection and, similar to AIM, ensures that all chosen marginal combinations do not exceed the predetermined maximum value. (The value indirectly limits the size of the model, preventing it from becoming too large and causing intractability.) We consider that each workload of the three datasets consists of a set of three-way marginal queries, and each attribute is uniformly and randomly selected.
Moreover, the randomness in the workload should be generated using a fixed random seed to ensure consistency of the workload across different datasets and mechanisms.

We can observe from the results that our approach PPSyn consistently outperforms other methods in most cases. Specifically, in the ADULT and BR2000 datasets, the frequency distribution of attribute values for each attribute is relatively uniform, and extreme counts are not prevalent. As a result, PPSyn is able to effectively reduce noise error through partitioning, which leads to improved private data synthesis.

On the other hand, in the LOANS dataset, some attributes have large domains, but their distribution is more skewed. This leads to a higher reconstruction error due to the assumption of a uniform distribution. In this scenario, the data-based method PrivMRF performs better than the workload-based methods on the LOANS dataset. However, our PPSyn still demonstrates superior performance compared to the workload-based methods AIM and MST.

Furthermore, the figures indicate that the error decreases as the privacy budget $\epsilon$ increases. The effectiveness of our method becomes more pronounced when the privacy budget is set at a lower level.

\subsubsection{Range Queries}

We conduct 3-dim and 5-dim range queris on the private synthesis data and reports the error in Figure~\ref{Fig:accuracy_R} and Figure~\ref{Fig:accuracy_R_5}, respectively. 
The 3-dimensional range queries generate all triple combinations from all attributes. Each attribute is assigned a weight sampled from a square exponential distribution. This sampling method is similar to AIM , where the randomness introduced during the construction of the workload is determined by a fixed random seed for different mechanisms. This ensures that the workload remains consistent across different mechanisms' executions. For all three datasets, we randomly generate 210 attribute combinations to form the final workload used to evaluate the range query accuracy of different mechanisms. 
Under the constraint of maximum workload, The 5-dimensional range queries consist of a total of 37 attribute combinations. 

\begin{figure*}
  \centering
    \subfigure[ADULT.]{
      \label{Fig:ADULT_R}
      \includegraphics[width=0.3\linewidth]{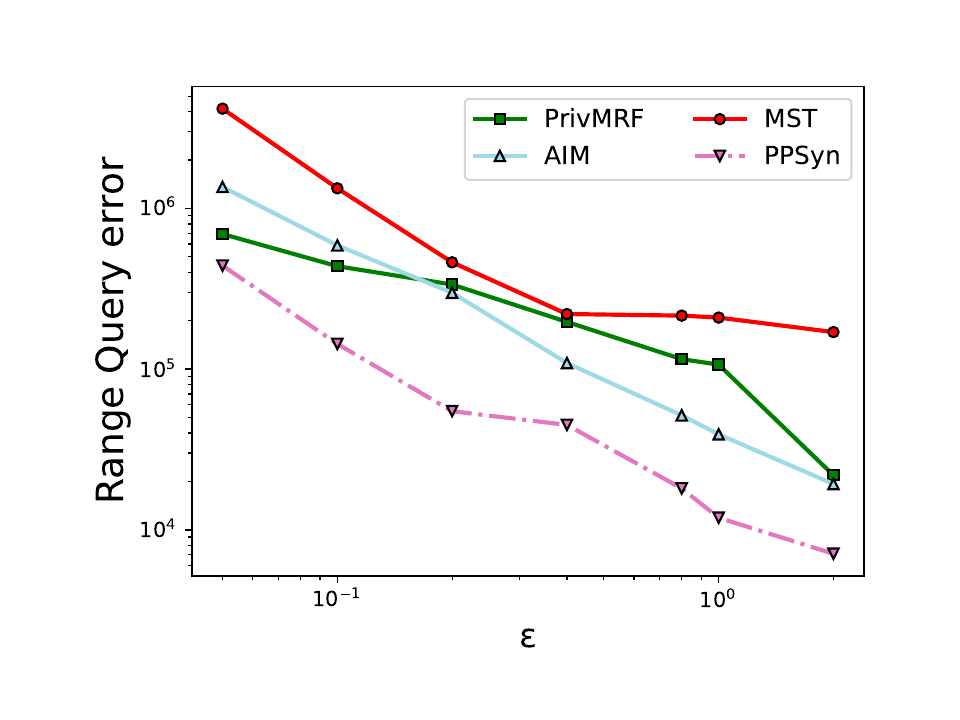}}
    \subfigure[BR2000.]{
      \label{Fig:BR2000_R}
      \includegraphics[width=0.3\linewidth]{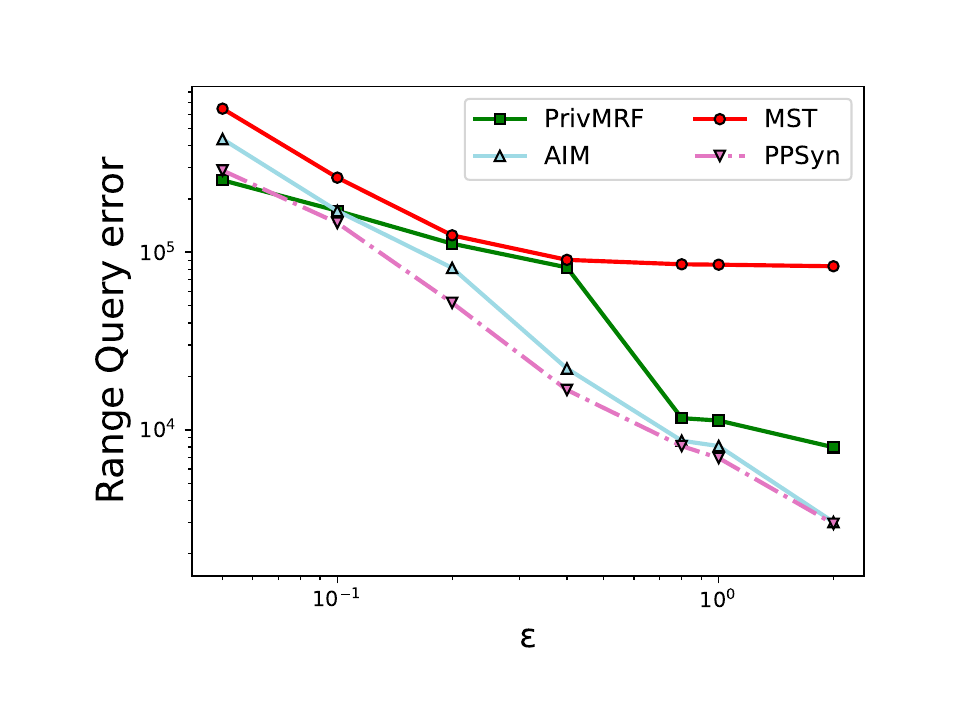}}
    \subfigure[LOANS.]{
      \label{Fig:LOAN_R}
      \includegraphics[width=0.3\linewidth]{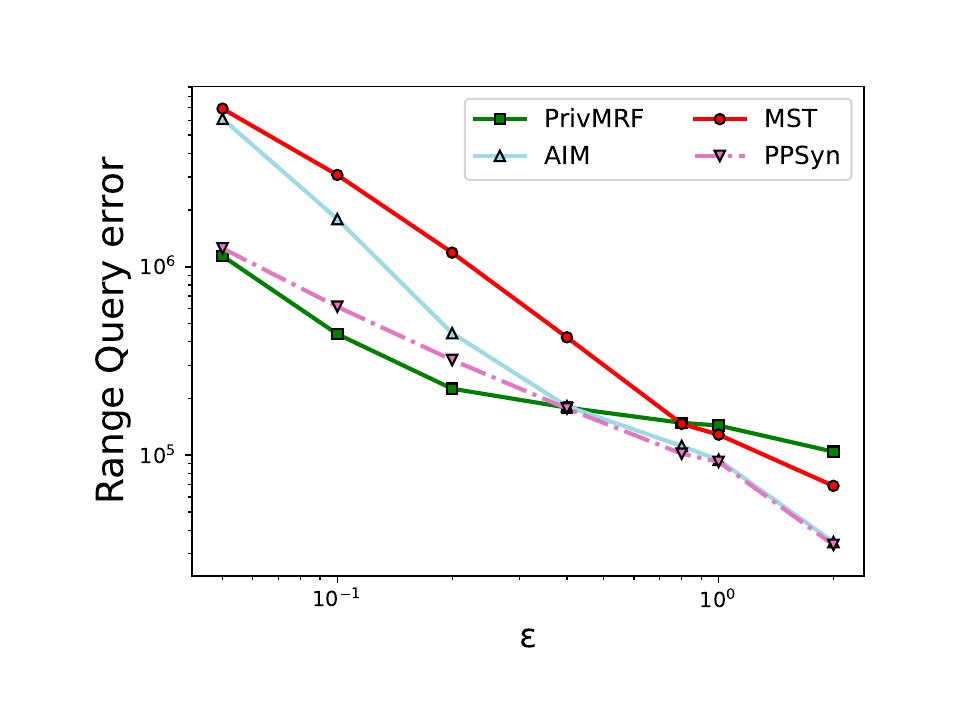}}
  \caption{The Accuracy of 3-Dim Range Queries}
  \label{Fig:accuracy_R}
  \end{figure*}

  \begin{figure*}
    \centering
      \subfigure[ADULT.]{
        \label{Fig:ADULT_R_5}
        \includegraphics[width=0.3\linewidth]{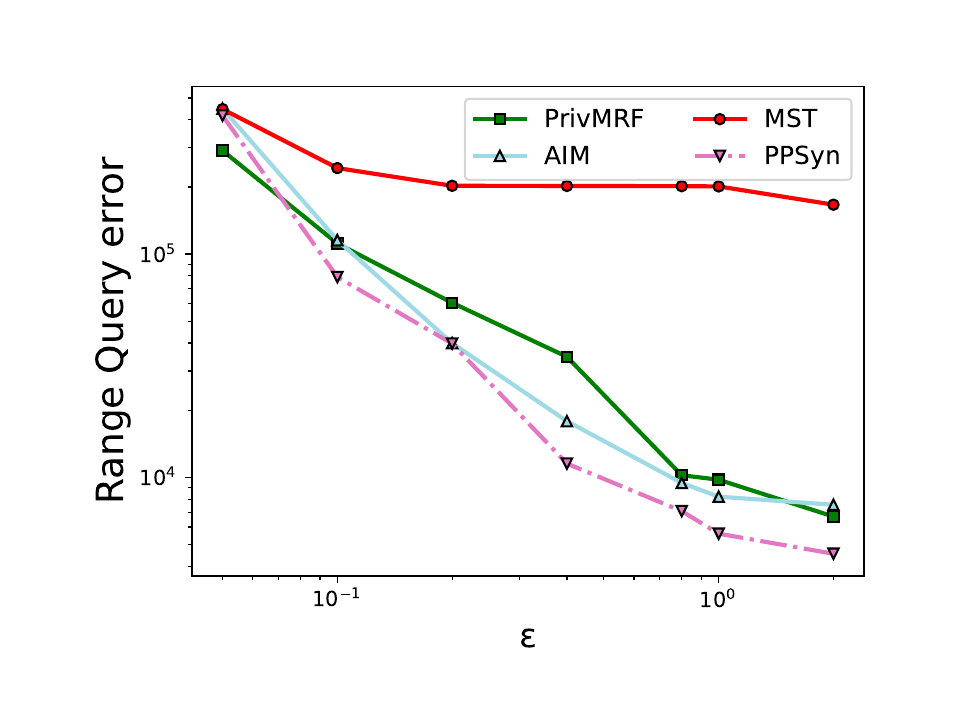}}
      \subfigure[BR2000.]{
        \label{Fig:BR2000_R_5}
        \includegraphics[width=0.3\linewidth]{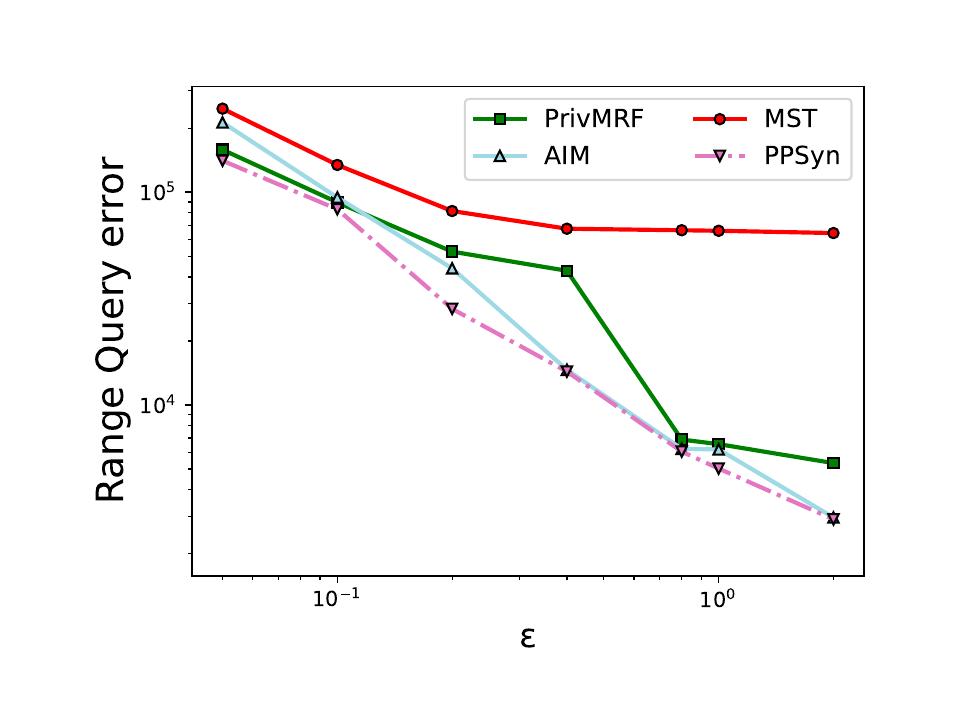}}
      \subfigure[LOANS.]{
        \label{Fig:LOAN_R_5}
        \includegraphics[width=0.3\linewidth]{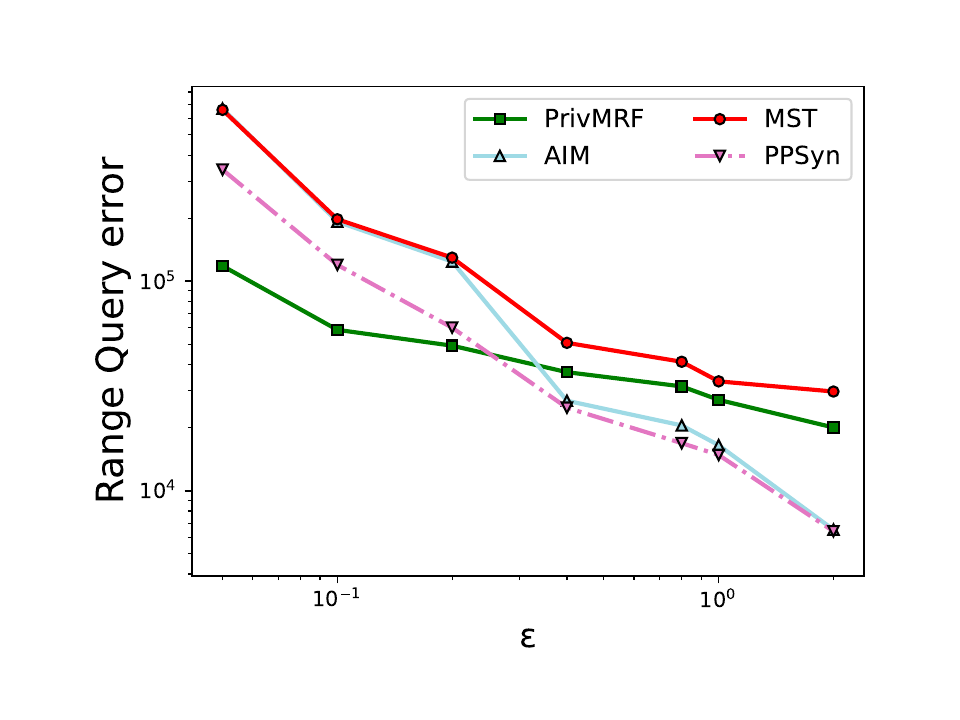}}
    \caption{The Accuracy of 5-Dim Range Queries}
    \label{Fig:accuracy_R_5}
    \end{figure*}
\subsection{Accuracy}

We consider privacy budget $\epsilon$ varying from 0.05 to 2.0. Figure~\ref{Fig:accuracy_R} shows the MSE comparison between our method and existing algorithms. Each plot corresponds to a dataset, with the x-axis representing changes in privacy budget, and the y-axis displaying the corresponding MSE results for range queries.
Our method achieves lower MSE compared to the other workload-based methods due to the improvement made in privacy budget allocation, which allows for more effective aggregation of valid counts during the range partitioning process.

We can observe that the data-based method sometimes outperforms other methods when handling range queries. This is because PrivMRF itself effectively utilizes the $\theta$-useful property to capture the correlations between attributes. For instance, in the loans dataset, many attributes exhibit extreme counts, with certain values having significantly more occurrences than others, resulting in sparse counts that introduce excessive noise. PrivMRF, during the marginals filtering stage, is expected to reduce the involvement of such attributes as much as possible to mitigate this issue. Meaning that, the data-based methods performs better of the queries no in the workload.

Under 5-dimensional range queries, PPSyn still demonstrates outstanding performance compared with the workload-based methods AIM and MST. In reality, for extensive range queries, it reduces the error introduced by the uniform distribution assumption. Because in this scenario, the queries' endpoints often ignore the uniform distribution assumption, which reduces the accumulation of errors stopping from such an assumption.

\begin{figure*}
  \centering
    \subfigure[ADULT.]{
      \label{Fig:ADULT_S}
      \includegraphics[width=0.45\linewidth]{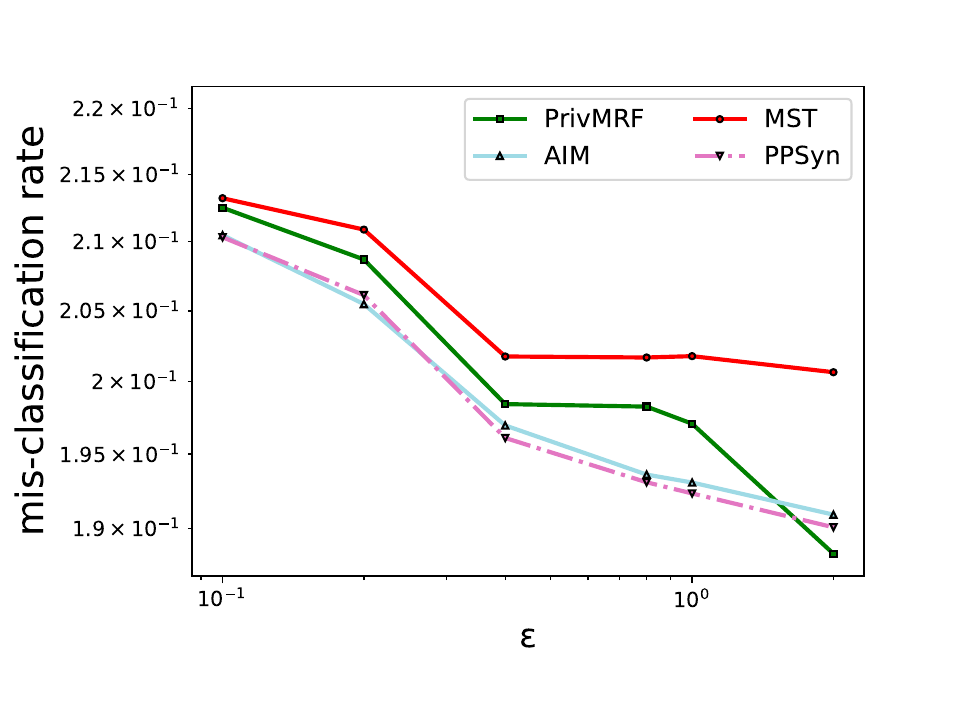}}
    \subfigure[BR2000.]{
      \label{Fig:BR2000_S}
      \includegraphics[width=0.45\linewidth]{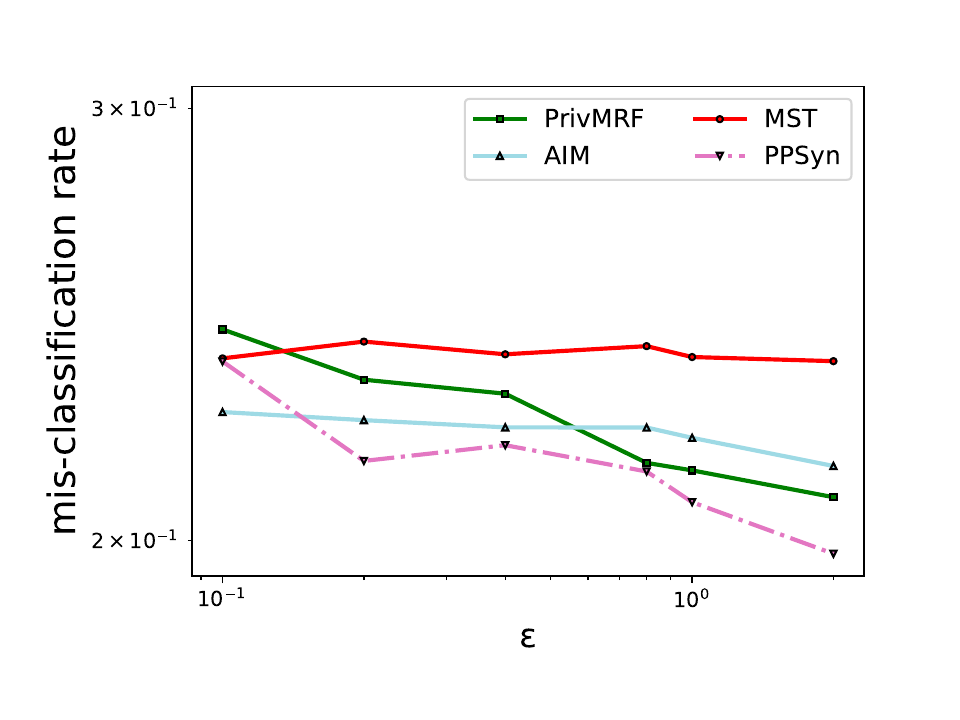}}

  \caption{The Accuracy of SVM}
  \label{Fig:accuracy_S}
  \end{figure*}

\subsubsection{Classification}

\textbf{Train the SVM classifier using the synthetic dataset.} 

For SVM classification, we followed the experimental settings used in the PrivMRF experiments to ensure consistency. In the training task, we employed 5-fold cross-validation. This means that we divided the dataset into 5 subsets, with 20\% of the data reserved as the test set and the remaining 80\% used as the training set for generating synthetic data for each mechanism.
To train the SVM classifier, we selected attributes of greater interest to users as the target attributes, while the remaining attributes served as features. The synthetic datasets generated by the privacy mechanisms were then used as the training set for the SVM classifier.
To evaluate the performance, we calculated the average misclassification rate across the 5 folds. This rate represents the proportion of misclassified instances in the test set.
We conducted these experiments on the ADULT and BR2000 datasets. However, we excluded the LOANS dataset from the SVM classification testing since it is not suitable for classification. The majority of attributes in the LOANS dataset are not suitable as target classification attributes. This would make it challenging to test whether the privacy mechanisms successfully retain the original features of these attributes during the classification process.

Based on the provided result in Figure~\ref{Fig:accuracy_S}, we can make the following observations.
 Our method consistently outperforms PrivMRF in terms of classification accuracy, especially when the privacy budget is low. This indicates that our method is more effective in generating synthetic data that is suitable for the SVM classification task. As the privacy budget increases, PrivMRF's performance in classification accuracy improves gradually. However, our method still maintains its superiority, showing a prominent improvement with increasing privacy budget.

These observations highlight the effectiveness of our method in generating synthetic data that preserves the original features necessary for SVM classification. Our method consistently outperforms other methods in terms of classification accuracy across different privacy budget settings. The reason for our method's superior performance is that we reduce the required budget of each large marginal, which allows us to allocate more budget to choose high-dimensional marginals that represent correlations among different dimensions. By leveraging these high-dimensional marginals, our method is able to capture and preserve complex relationships between attributes, leading to more accurate classification results.
In contrast, other methods such as PrivMRF may allocate a significant portion of the privacy budget to individual large marginals, limiting its ability to capture and preserve high-dimensional correlations. This can result in lower classification accuracy, especially in scenarios with low privacy budgets.

\subsection{Impact of parameters}

\begin{figure*}
  \centering
    \subfigure[ADULT.]{
      \label{Fig:ADULT_eta}
      \includegraphics[width=0.3\linewidth]{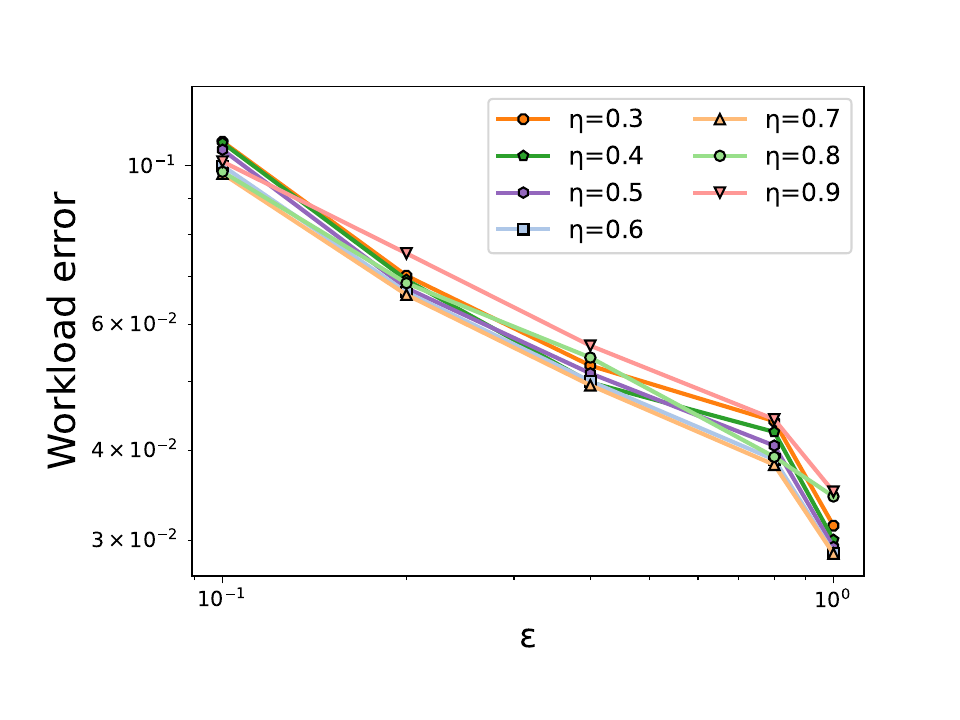}}
    \subfigure[BR2000.]{
      \label{Fig:BR2000_eta}
      \includegraphics[width=0.3\linewidth]{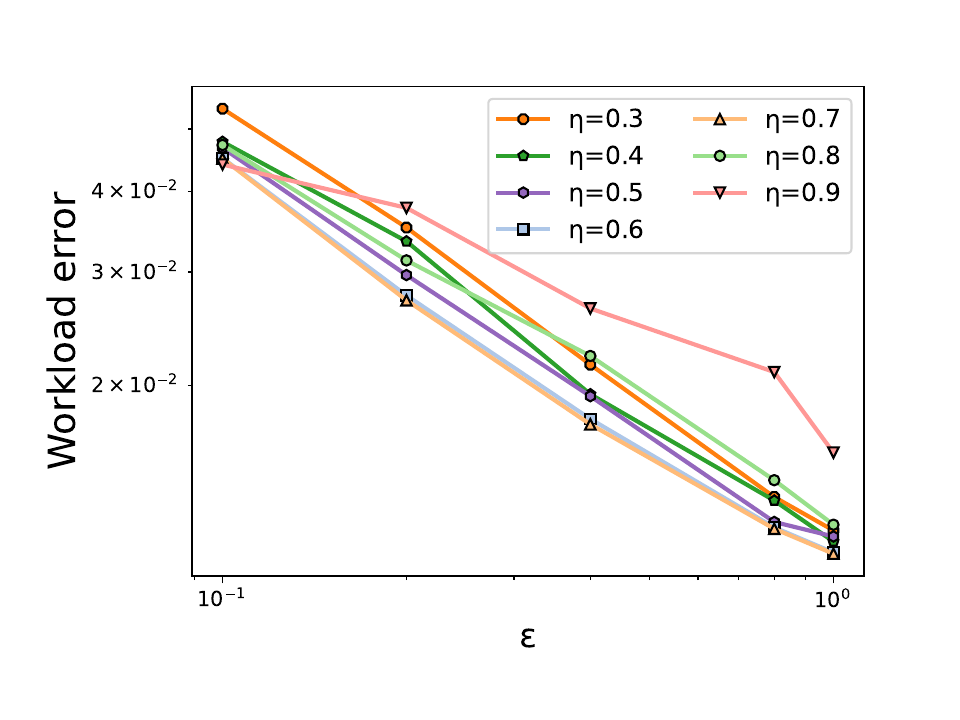}}
    \subfigure[LOANS.]{
      \label{Fig:LOAN_eta}
      \includegraphics[width=0.3\linewidth]{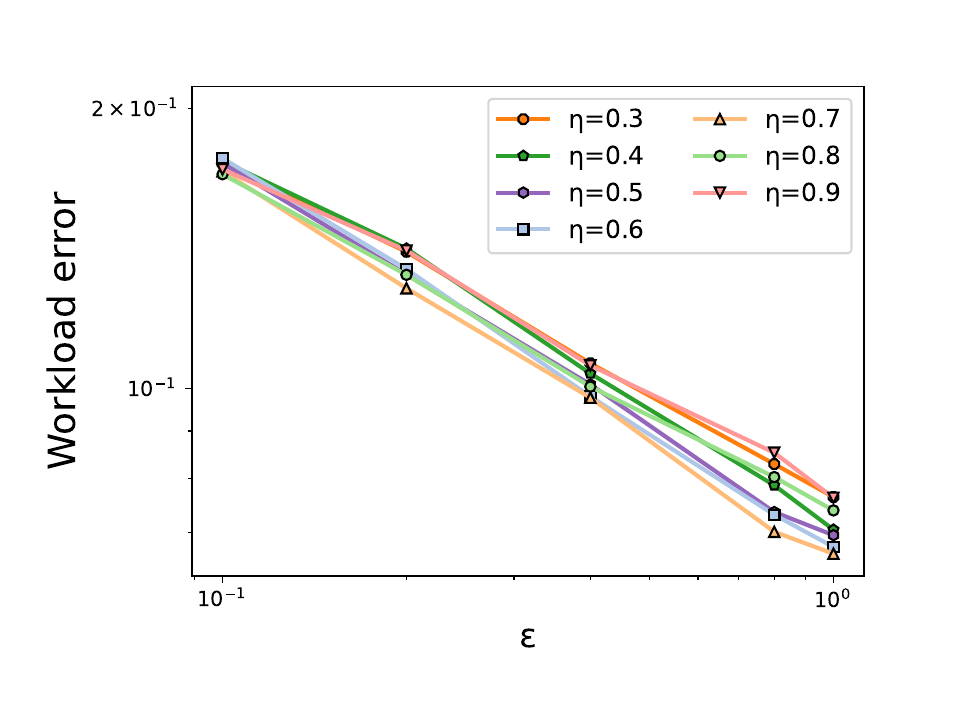}}
  \caption{The impact of $\eta$ on workload error}
  \label{Fig:accuracy_eta}
  \end{figure*}

  \begin{figure*}
    \centering
      \subfigure[ADULT.]{
        \label{Fig:ADULT_R_eta}
        \includegraphics[width=0.3\linewidth]{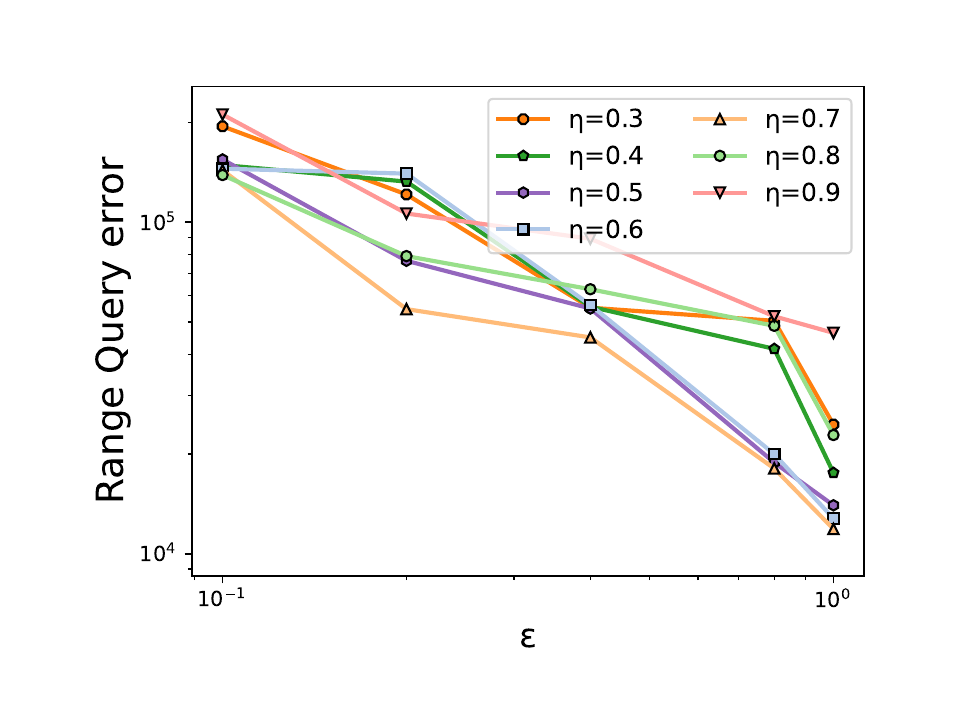}}
      \subfigure[BR2000.]{
        \label{Fig:BR2000_R_eta}
        \includegraphics[width=0.3\linewidth]{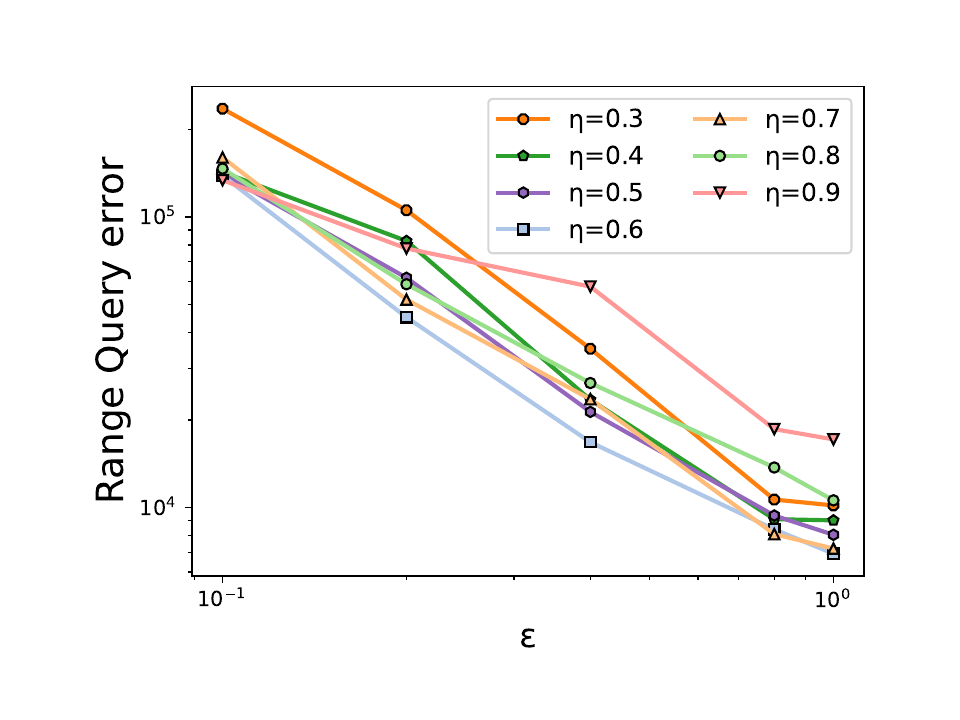}}
      \subfigure[LOANS.]{
        \label{Fig:LOAN_R_eta}
        \includegraphics[width=0.3\linewidth]{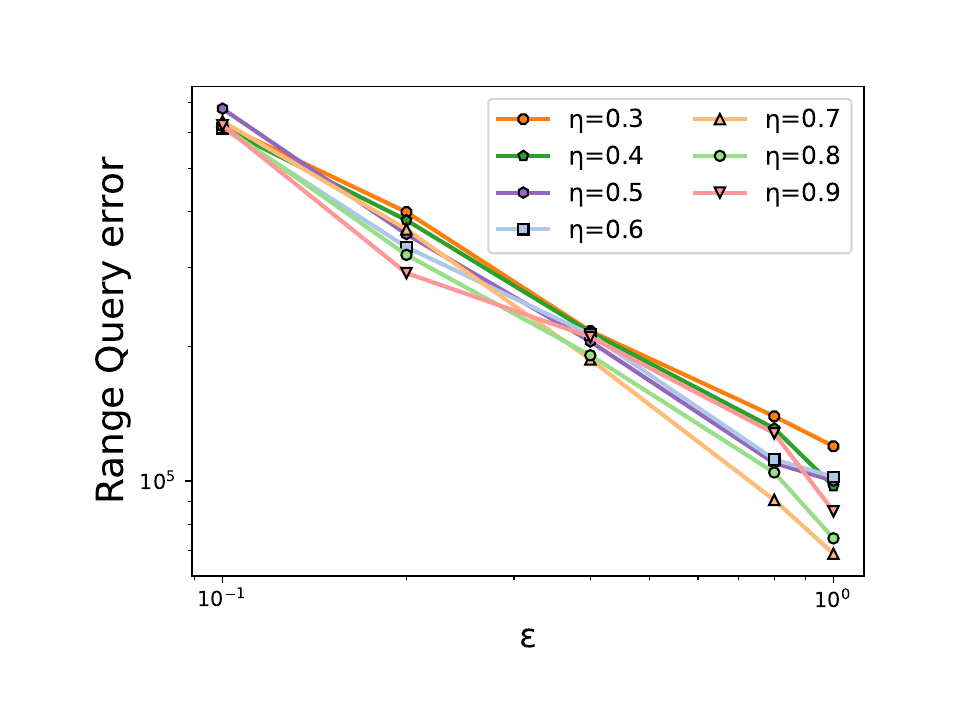}}
    \caption{The impact of $\eta$ on range query error}
    \label{Fig:accuracy_R_eta}
    \end{figure*}

In this section, we will consider the impact of $\eta$ and $d$ to PPSyn. We will examine the effects on both workload error and range query, with a focus on workload error for parameter $d$.

\textbf{Impact of $\eta$ on the accuracy.} 
We will evaluate the influence of $\eta$  by varying it from 0.3 to 1.0 with intervals of 0.1. Intuitively, marginals with true errors ranging from 0.7 to 0.8 are considered valuable, as they effectively preserve important information while introducing acceptable levels of noise. Marginals close to 1.0, however, may carry excessive error, leading to the accumulation of unnecessary errors.

As depicted in Figure~\ref{Fig:accuracy_eta} and Figure~\ref{Fig:accuracy_R_eta}, when the error parameter $\eta$ is set close to 0.7, it is possible to effectively reduce the impact of errors. In our experimental observations, we find that under low error control, the number of selectable marginals decreases. On the other hand, under high error control, the algorithm accumulates marginals where noise overwhelms valuable information. Therefore, choosing an appropriate error control value can help prevent the accumulation of marginals that are filled with excessive error.

Another interesting observation is that as the $\eta$ parameter approaches 1.0, more marginals participate in the calculation of the probabilistic graph. This is because the budget allocated to each marginal decreases, leading to more selection rounds. However, as a consequence, the accumulation of noise becomes more pronounced, resulting in diminished final accuracy. This emphasizes the importance of reasonable error control.

By carefully selecting the value of $\eta$, we can effectively manage the impact of error on the contribution of each marginal. This ensures that PPSyn keeps selecting marginals bringing positive contribution to the distribution model.




\begin{figure*}
  \centering
      \includegraphics[width=0.6\linewidth]{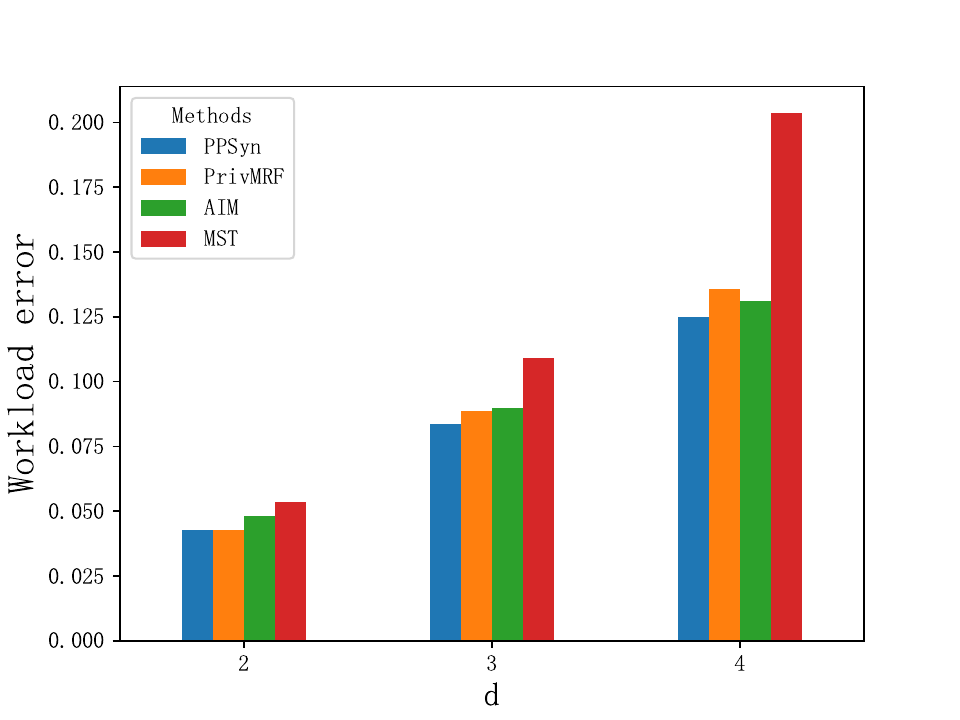}
  \caption{The impact of $d$ on workload error (BR2000, $\epsilon$=0.1)}
  \label{Fig:accuracy_W_3}
\end{figure*}

\textbf{Impact of  $d$ on the accuracy.} 
    The parameter $d$ represents the dimensionality of each marginal in the workload. Figure~\ref{Fig:accuracy_W_3} compares all methods in terms of the workload error of the $d$-way marginals derived from the synthetic data on BR2000, with a privacy budget set at 0.1. As expected, all mechanisms experience an increase in error as $d$ increases. This is because the increase in zero-count attribute values leads to more counts being overwhelmed by noise. However, we can see that our method exhibits the lowest error in the figures. This can be attributed to our proposed partition-based method, which reduces the occurrence of zero-count items. Another interesting observation is that PrivMRF starts performing worse than the workload-based methods when $d$=4. This could potentially be due to the limitation of using only a small number of low-dimensional marginals for high-dimensional distribution estimation. Additionally, as $d$ increases, the gap between PPSyn and other mechanisms becomes even more pronounced.

\subsection{Summaries of experimental results.}

(1) PPSyn is more accurate than the state-of-the-art workload-based methods including AIM and MST for both range queries and SVM classification tasks.
  
(2) PPSyn even outperforms the data-based method PrivMRF on the dataset with large domain.
  
(3) PPSyn is more suitable than the other workload-based methods AIM and MST for workloads with multidimensional marginals.

\section{Conclusion}
In this paper, we introduce a novel algorithm called PPSyn for generating synthetic data. Our approach focuses on reducing the privacy budget required for each marginal while ensuring its positive contribution to the overall representation of the multidimensional distribution. By leveraging partition methods, we are able to allocate the privacy budget more efficiently, enabling us to incorporate larger marginals into the synthetic data generation process.
Compared to the state-of-the-art workload-based private data synthesis methods, PPSyn demonstrates superior performance when the privacy budget is limited at a lower level. Moving forward, we plan to extend our approach to address synthetic data generation under local differential privacy.

\section*{Acknowledgements}
This research was partially supported by the NSFC grant 62202113; GuangDong Basic and Applied Basic Research Foundation SL2022A04J01306; Open Project of Jiangsu Province Big Data Intelligent Engineering Laboratory SDGC2229; the Major Key Project of PCL (Grant No.PCL2021A09, PCL2021A02, PCL2022A03).



\bibliographystyle{elsarticle-num}
\bibliography{mybibfile}
\end{document}